\documentclass[11pt,a4paper,notcite,notref]{article}

\usepackage{amssymb,amsmath}
\usepackage{tipa}
\usepackage{epic}
\usepackage{amsfonts}
\usepackage{accents}
\usepackage{texdraw}
\usepackage{color}
\usepackage{eucal}
\usepackage{cite}
\usepackage[all]{xy}
\usepackage{bm}
\usepackage{graphics,amsmath,amssymb,amsthm,accents}
\def \h#1{\widehat{#1}}
\def \t#1{\widetilde{#1}}

\newtheorem{lem}{Lemma}
\newtheorem{Theorem}{Theorem}
\newtheorem{Proposition}{Proposition}
\numberwithin{equation}{section}

\def \tyb#1{\hbox{\tiny{[{\it{#1}}]}}}
\def \ty#1{\hbox{\tiny{{\it{#1}}}}}

\def \h#1{\widehat{#1}}
\def \t#1{\widetilde{#1}}

\def \th#1{\widehat{\widetilde{#1}}}

\DeclareMathAccent{\wtilde}{\mathord}{largesymbols}{"65}
\DeclareMathAccent{\what}{\mathord}{largesymbols}{"62}

\def\m@th{\mathsurround=0pt}
\mathchardef\bracell="0365
\def\upbrall{$\m@th\bracell$}
\def\undertilde#1{\mathop{\vtop{\ialign{##\crcr
    $\hfil\displaystyle{#1}\hfil$\crcr
     \noalign
     {\kern1.5pt\nointerlineskip}
     \upbrall\crcr\noalign{\kern1pt
   }}}}\limits}

\def\underhat#1{\mathop{\vtop{\ialign{##\crcr
    $\hfil\displaystyle{#1}\hfil$\crcr
     \noalign
     {\kern1.5pt\nointerlineskip}
     \upbrall\crcr\noalign{\kern1pt
   }}}}\limits}

\newcommand{\wh}{\widehat}
\newcommand{\wt}{\widetilde}

\newcommand{\tc}{\,^t \hskip -2pt {\boldsymbol{c}}}

\newcommand{\Ga}{\boldsymbol{\Gamma}}

\newcommand{\bA}{\boldsymbol{A}}
\newcommand{\bB}{\boldsymbol{B}}

\newcommand{\bF}{\boldsymbol{F}}
\newcommand{\bG}{\boldsymbol{G}}
\newcommand{\bH}{\boldsymbol{H}}
\newcommand{\bI}{\boldsymbol{I}}

\newcommand{\bK}{\boldsymbol{K}}

\newcommand{\bM}{\boldsymbol{M}}

\newcommand{\bT}{\boldsymbol{T}}

\newcommand{\bc}{\boldsymbol{c}}
\newcommand{\br}{{\boldsymbol{r}}}
\newcommand{\bu}{\boldsymbol{u}}

\title{Solutions to the ABS lattice equations via\\
generalized Cauchy matrix approach}

\author{Da-jun Zhang\footnote{Corresponding author. E-mail address: djzhang@staff.shu.edu.cn},
~~Song-lin Zhao
\vspace{4mm}\\
{\small\it Department of Mathematics, Shanghai University, Shanghai 200444, P.R. China}}

\date{\today}

\begin{document}

\maketitle

\begin{abstract}
The usual Cauchy matrix approach starts from a known plain wave factor vector $\br$ and known dressed Cauchy matrix $\bM$.
In this paper we start from a matrix equation set with undetermined $\br$ and $\bM$.
From the starting equation set we can build shift relations for some defined scalar functions and then derive lattice equations.
The starting  matrix equation set admits more choices for $\br$ and $\bM$
and in the paper we give explicit formulae for all possible $\br$ and $\bM$.
As applications, we get more solutions than usual multi-soliton solutions
for many lattice equations
including  the lattice potential KdV  equation,
the lattice potential modified KdV  equation,
the lattice Schwarzian  KdV  equation, NQC equation
and some lattice equations in ABS list.




\end{abstract}


\section{Introduction}

In recent decades the research of integrability of difference equations
has got remarkable progress.
The property of multi-dimensional consistency\cite{Nijhoff-MDC,BS-2002,ABS-2003}
reveals integrability in some sense for discrete equations.
Based on this property lattice equation defined on an elementary square
can be classified\cite{ABS-2003} and the result is referred to as ABS list,
which is surprisingly short and only consists of nine equations named as
Q4, Q3, Q2, Q1, A2, A1, H3, H2 and H1.
We list out these equations in Appendix \ref{A:2}.

ABS list has received a lot of attention.
With regard to solutions,  a lattice equation that is multi-dimensionally consistent
provides automatically a B\"{a}cklund transformation\cite{ABS-2003}.
Then, by the B\"{a}cklund transformations one can derive both seed solutions and 1-soliton solutions\cite{Soliton-Q4,Hietarinta-ABS}.
For $N$-soliton solutions, the relation between Q3 equation and NQC equation\cite{Soliton-Q3,Nijhoff-ABS}
can be used.
NQC equation is derived in Cauchy matrix approach\cite{Nijhoff-ABS,NQC-1983}
and its solution can be expressed through Cauchy matrices.
On the other hand, Q3 equation can be degenerated to the lower equations Q2, Q1, H3, H2 and H1.
It then follows that $N$-soliton solutions of these equations can be given in terms of Cauchy matrices\cite{Nijhoff-ABS}.
Besides Cauchy matrix approach, several other approaches are developed for finding multi-soliton solutions of the lattice equations in ABS list,
such as bilinear approach\cite{Hietarinta-ABS}, transformation and iteration method\cite{Atkinson-2009},
Cauchy matrix approach with elliptic functions\cite{Nijhoff-ABS-E},
Inverse Scattering Transform\cite{H1-IST,Q3-IST}, and algebro-geometry approach\cite{Cao-1,Cao-2}.

The Cauchy matrix approach\cite{Nijhoff-ABS}(also see \cite{NQC-1983,Quispel-1984,math5492}),
starts from a known plain wave factor vector $\br$ and known dressed Cauchy matrix $\bM$, defines some scalar dependent variables,
constructs their recurrence relations,
and then from closed relations provides discrete lattice equations,
such as the lattice potential KdV (lpKdV) equation,
the lattice potential modified KdV (lpmKdV) equation,
the lattice Schwarzian  KdV (lSKdV) equation and NQC equation.
Solutions of these obtained lattice equations can be given through the dressed Cauchy matrix.

In this paper, instead of known plain wave factor vector $\br$ and known dressed Cauchy matrix $\bM$,
we start from a matrix equation set consisting of three equations,
among which two equations are used to determine plain wave factor vector $\br$ and the third equation is
used to define $\bM$.
We can build shift relations for some defined scalar functions and then derive lattice equations.
This procedure we call generalized Cauchy matrix approach.
In fact, the starting  matrix equation set is demonstrated to admit more choices for $\br$ and $\bM$ which leads to more solutions
than usual solitons.
In the paper we will give explicit forms of all possible solutions to the starting  matrix equation set.
As applications, these solutions are used to construct solutions for many
lattice equations,
such as the lpKdV, lpmKdV, lSKdV, NQC, Q3, Q2, Q1, H3, H2 and H1 equation.

The paper is organized as follows. In Sec. 2, we briefly review the Cauchy matrix approach.
In Sec.3 we describe the generalized Cauchy matrix approach.
In Sec.4 we solve the starting  matrix equation set and in Sec.5
as applications solutions for some lattice equations are given.

\section{Cauchy matrix approach}
\label{sec:2}

As a preliminary part let us briefly review the Cauchy matrix approach.
For more details one can refer to Ref.\cite{Nijhoff-ABS} or  \cite{math5492}.

A Cauchy matrix is known as a square matrix
$\bG=(G_{i,j})_{N\times N},~G_{i,j}=\frac{1}{w_j-z_i}$, where $z_i\neq w_j\in \mathbb{C}$.
We will use its symmetric form, which is
\begin{equation}
\bG=(G_{i,j})_{N\times N},~~G_{i,j}=\frac{1}{k_i+k_j},~~k_i\neq -k_j\in \mathbb{C}.
\end{equation}
The Cauchy matrix approach starts from a ``dressed'' Cauchy matrix
\begin{equation}
\bM=(M_{i,j})_{N\times N}=\bF\times \bG\times \bH,~~M_{i,j}=\frac{\rho_i c_j}{k_i+k_j},
\end{equation}
where $c_i\in \mathbb{C}$, $\rho_i$ is called the plain wave factor defined as
\begin{equation}
\rho_i=\biggl(\frac{p +k_i}{p -k_i}\biggr)^n\biggl(\frac{q + k_i}{q-k_i}\biggr)^m \rho^0_{i}
\end{equation}
with constants $p,q,\rho^0_{i}$,
and the dressing matrices are
\[
\bF=\mathrm{Diag}(\rho_1,\cdots,\rho_N), ~~~
\bH=\mathrm{Diag}(c_1,\cdots,c_N).
\]
$\bM$ satisfies the relation
\begin{equation}
\bM \bK+ \bK \bM = {\br}\, \tc,
\label{rel-MK}
\end{equation}
where
\begin{equation}
\bK=\mathrm{Diag}(k_1,k_2,\cdots,k_N),~~
\br=(\rho_1,\rho_2,\cdots,\rho_N)^T, ~~
\tc=(c_1,c_2,\cdots,c_N).
\end{equation}
By $\t{}$ and $\h{}$ we respectively denote the shifts in $n$ and $m$ direction, i.e.,
$\t f(n,m)=f(n+1,m)$, $\h f(n,m)=f(n,m+1)$.
Then, from the basic shift relation
\[\t \rho_i=\frac{p+k_i}{p-k_i}\rho_i,~~\h \rho_i=\frac{q+k_i}{q-k_i}\rho_i,\]
one may build the following shift relations of $\bM$,
\begin{subequations}
\begin{eqnarray}
&& \wt{\bM}(p\bI+\bK)-(p\bI+\bK)\bM=\wt{\br}~\tc, \label{eq:M-dyna-1}\\
&& (p\bI-\bK)\wt{\bM}-\bM(p\bI-\bK)=\br~\tc, \label{eq:M-dyna-2}\\
&& \wh{\bM}(q\bI+\bK)-(q\bI+\bK)\bM=\wh{\br}~\tc,\label{eq:M-dyna-3}\\
&& (q\bI-\bK)\wh{\bM}-\bM(q\bI-\bK)=\br~\tc,\label{eq:M-dyna-4}
\end{eqnarray}
\label{M-dyna}
\end{subequations}
where $\bI$ is the $N\times N$ unit matrix.

Next, for any $(i,j)\in \mathbb{Z}\times \mathbb{Z}$, introduce a scalar function
\begin{equation}
S^{(i,j)}= \tc \,\bK^j(\bI+ \bM)^{-1} \bK^i \br,
\label{Sij-1}
\end{equation}
which can be proved to have the symmetric property
\begin{equation}
S^{(i,j)}=S^{(j,i)}.
\label{Sij-1-sym}
\end{equation}
It then follows from the dynamical relation \eqref{M-dyna} that one can reach to a set of
recurrence relations:
\begin{subequations}
\begin{eqnarray}
&& p \wt{S}^{(i,j)}-\wt{S}^{(i,j+1)}=p S^{(i,j)}+S^{(i+1,j)}-\wt{S}^{(i,0)}S^{(0,j)},\label{eq:Sij-dyna-1}\\
&& p S^{(i,j)}+S^{(i,j+1)}=p \wt{S}^{(i,j)}-\wt{S}^{(i+1,j)}+S^{(i,0)}\wt{S}^{(0,j)},\label{eq:Sij-dyna-2}\\
&& q \wh{S}^{(i,j)}-\wh{S}^{(i,j+1)}=q S^{(i,j)}+S^{(i+1,j)}-\wh{S}^{(i,0)}S^{(0,j)},\label{eq:Sij-dyna-3}\\
&& q S^{(i,j)}+S^{(i,j+1)}=q
\wh{S}^{(i,j)}-\wh{S}^{(i+1,j)}+S^{(i,0)}\wh{S}^{(0,j)}.\label{eq:Sij-dyna-4}
\end{eqnarray}
\label{Sij-1-dyna}
\end{subequations}

Besides $S^{(i,j)}$, there is another scalar function defined by
\begin{equation}
S(a,b)=\tc(b\bI+\bK)^{-1}(\bI+\bM)^{-1}(a\bI+\bK)^{-1}\br,~~ a,b\in \mathbb{C},
\label{Sab-1}
\end{equation}
which also has symmetric property
\begin{equation}
S(a,b)=S(b,a).
\label{Sab-sym}
\end{equation}
It can be proved that $S(a,b)$ obeys the shift relations
\begin{subequations}
\label{Sab-1-re}
\begin{eqnarray}
&& 1-(p+b)\wt{S}(a,b)+(p-a)S(a,b)=\wt{V}(a)V(b), \label{Sab-1-re-a}\\
&& 1-(q+b)\wh{S}(a,b)+(q-a)S(a,b)=\wh{V}(a)V(b), \label{Sab-1-re-b}
\end{eqnarray}
\end{subequations}
where
\begin{eqnarray}
V(a)=1-\tc(a \bI+ \bK)^{-1}(\bI+\bM)^{-1}\br=1-\tc (\bI+\bM)^{-1}(a\bI+ \bK)^{-1}\br.
\label{Va}
\end{eqnarray}

With the help of the crucial symmetric properties \eqref{Sij-1-sym} and \eqref{Sab-sym},
some lattice equations appear as closed forms of the recurrence relations \eqref{Sij-1-dyna} and \eqref{Sab-1-re}.
We list those equations in Sec.\ref{sec:3.2}.

Since $S^{(i,j)}$ and $S(a,b)$ are  defined by the known elements $\bM, \bK, \br, \tc$, solutions of
those lattice equations are therefore given apparently.

\section{Generalized Cauchy matrix approach}
\label{sec:3}

In the previous section we briefly introduced the Cauchy matrix approach.
Following this approach several lattice equations can be derived and their solutions are expressed through
the known elements including
the dressed Cauchy matrix $\bM$, plain wave factor vector $\br$, constant diagonal matrix $\bK$ and constant vector $\tc$.

In the following, we start with unknown $\bM$, $\br$ and $\bK$,
and investigate a generalized Cauchy matrix approach.

\subsection{Recurrence relations}

Let us first give the following Lemma.

\begin{lem}\label{L:1-1}
Suppose that matrices  $\bK, \bA\in \mathbb{C}_{N\times N}$ are anticommutative, i.e.,
\begin{equation}
\bK \bA+\bA \bK=0,
\label{anti-c}
\end{equation}
where all the eigenvalues $\{k_1,k_2,\cdots, k_N\}$ (some of them can be same) of $\bK$
satisfy
\begin{equation}
k_i+k_j\neq 0,~~ \forall~ 1\leq i,j\leq N.
\label{kij}
\end{equation}
Then $\bA$ is a zero matrix.
\end{lem}

We leave the proof in Appendix \ref{A:1}.
Next let us start to derive some recurrence relations.

\begin{Theorem}
\label{T:1}
Consider the $N\times N$ matrices $\bM=(M_{i,j}(n,m))_{N\times N}$ and $\bK=(K_{i,j})_{N\times N}$,
and the $N$-th order vectors $\br=(\rho_1(n,m),\rho_2(n,m),\cdots,\rho_N(n,m))^T$ and $\tc=(c_1,c_2,\cdots,c_N)$,
where $M_{i,j}(n,m)$ and $\rho_j(n,m)$ are undetermined functions while $K_{i,j}$ and $c_j$ are constants.
Suppose that $\bM$, $\br$,  $\bK$ and  $\tc$ obey the relations
\begin{subequations}
\label{cond}
\begin{align}
& (p \bI- \bK)\wt{\br}  =(p \bI+\bK) \br,\label{cond-r-p}\\
& (q \bI-\bK)\wh{\br}  =(q \bI+\bK) \br,\label{cond-r-q}\\
& \bM \bK+ \bK \bM = \br \tc, \label{cond-MK}
\end{align}
\end{subequations}
which we call the starting matrix equation set,
where $p,q$ are constants, $\bI$ is the $N\times N$ unit matrix.
Define scalar function
\begin{equation}
S^{(i,j)}= \tc \,\bK^j(\bI+ \bM)^{-1} \bK^i \br,~~i,j\in \mathbb{Z},
\label{Sij-2}
\end{equation}
then we have the following recurrence relations:
\begin{subequations}\label{Sij-shift}
\begin{eqnarray}
&& p \wt{S}^{(i,j)}-\wt{S}^{(i,j+1)}=p S^{(i,j)}+S^{(i+1,j)}-\wt{S}^{(i,0)}S^{(0,j)},\label{Sij-shift-a}\\
&& p S^{(i,j)}+S^{(i,j+1)}=p \wt{S}^{(i,j)}-\wt{S}^{(i+1,j)}+S^{(i,0)}\wt{S}^{(0,j)},\label{Sij-shift-b}\\
&& q \wh{S}^{(i,j)}-\wh{S}^{(i,j+1)}=q S^{(i,j)}+S^{(i+1,j)}-\wh{S}^{(i,0)}S^{(0,j)},\label{Sij-shift-c}\\
&& q S^{(i,j)}+S^{(i,j+1)}=q
\wh{S}^{(i,j)}-\wh{S}^{(i+1,j)}+S^{(i,0)}\wh{S}^{(0,j)}.\label{Sij-shift-4}
\end{eqnarray}
\end{subequations}
Define another scalar function
\begin{equation}
S(a,b)=\tc(b\bI+\bK)^{-1}(\bI+\bM)^{-1}(a\bI+\bK)^{-1}\br,~~ a,b\in \mathbb{C}.
\label{Sab-2}
\end{equation}
If $S^{(i,j)}$ satisfies the symmetric property
\begin{equation}
S^{(i,j)}=S^{(j,i)},
\label{Sij-sym}
\end{equation}
then $S(a,b)$ is also of symmetric form
\begin{equation}
S(a,b)=S(b,a)
\end{equation}
and satisfies  the shift relations:
\begin{subequations}
\label{Sab-2-re}
\begin{eqnarray}
&& 1-(p+b)\wt{S}(a,b)+(p-a)S(a,b)=\wt{V}(a)V(b), \label{Sab-2-re-a}\\
&& 1-(q+b)\wh{S}(a,b)+(q-a)S(a,b)=\wh{V}(a)V(b), \label{Sab-2-re-b}\\
&& 1-(p+a)\wt{S}(a,b)+(p-b)S(a,b)=\wt{V}(b)V(a), \label{Sab-2-re-c}\\
&& 1-(q+a)\wh{S}(a,b)+(q-b)S(a,b)=\wh{V}(b)V(a), \label{Sab-2-re-d}
\end{eqnarray}
\end{subequations}
where
\begin{equation}
V(a)=1-\tc (\bI+\bM)^{-1}(a\bI+ \bK)^{-1}\br.
\label{Va-def}
\end{equation}
Here in the following we suppose $s\bI\pm \bK$ is invertible for $s=0,p,q,a,b$,
and all the eigenvalues $\{k_1,k_2,\cdots, k_N\}$ (some of them can be same) of $\bK$
satisfy
\begin{equation}
k_i+k_j\neq 0,~~ \forall~ 1\leq i,j\leq N.
\label{kij-2}
\end{equation}
\end{Theorem}

\begin{proof}
Let us first derive a shift relation of $\bK$ and $\bM$.
From $\t{\eqref{cond-MK}}$ and \eqref{cond-r-p} and noting that $(p\bI\pm \bK)\bK=\bK(p\bI\pm \bK)$,
we have
\begin{equation}
(p\bI-\bK)\t{\bM}\bK+\bK(p\bI-\bK)\t\bM=(p\bI+\bK)\br \tc.
\label{t1}
\end{equation}
Meanwhile, just left-multiplying $p\bI+\bK$ on \eqref{cond-MK} yields
\begin{equation}
(p\bI+\bK){\bM}\bK+\bK(p\bI+\bK)\bM=(p\bI+\bK)\br \tc.
\label{t2}
\end{equation}
Subtracting \eqref{t2} from \eqref{t1} yields
\[[(p \bI-\bK) \t{\bM}-(p \bI+\bK)\bM]\bK+\bK[(p \bI-\bK) \t{\bM}-(p \bI+\bK)\bM]=0,\]
which further, in the light of Lemma \ref{L:1-1}, gives
\begin{subequations}\label{cond-M-th}
\begin{equation}
(p \bI-\bK) \wt{\bM}=(p \bI+\bK)\bM. \label{cond-M-t}
\end{equation}
For $q$ and hat-shift we have
\begin{equation}
(q \bI- \bK) \wh{\bM}=(q \bI+\bK)  \bM. \label{cond-M-h}
\end{equation}
\end{subequations}

In the following, the proof is actually similar to the usual Cauchy matrix approach\cite{Nijhoff-ABS}.
Taking $\wt{\phantom{a}}$ shift of \eqref{cond-MK} we have
\[
\wt{\br}\, \tc=\wt{\bM}\bK+\bK\wt{\bM},
\]
and replacing the last term $\bK\wt{\bM}$ by \eqref{cond-M-t} yields
\begin{subequations}\label{eq:M-dyna-g}
\begin{eqnarray}
\wt{\bM}(p\bI+\bK)-(p\bI+\bK)\bM=\wt{\br}\, \tc. \label{eq:M-dyna-1}
\end{eqnarray}
Besides, if deleting the term $\bK \bM$ from \eqref{cond-MK} and \eqref{cond-M-t} we have
\begin{equation}
(p\bI-\bK)\wt{\bM}-\bM(p\bI-\bK)=\br\, \tc. \label{eq:M-dyna-2}
\end{equation}
Similarly, using \eqref{cond-MK} and \eqref{cond-M-h} we have
\begin{eqnarray}
&& \wh{\bM}(q\bI+\bK)-(q\bI+\bK)\bM=\wh{\br}\,\tc,\label{eq:M-dyna-3}\\
&& (q\bI-\bK)\wh{\bM}-\bM(q\bI-\bK)=\br\,\tc.\label{eq:M-dyna-4}
\end{eqnarray}
\end{subequations}

Now, noting that the above shift relation \eqref{eq:M-dyna-g} is exactly the same as \eqref{M-dyna},
and the matrices $\bK, ~(a\bI+\bK)$ and $(b\bI+\bK)^{-1}$ commute each other as being diagonals,
one can get the recurrence relation \eqref{Sij-shift} as in \cite{Nijhoff-ABS}.
For completeness, we show the procedure in the following.
We introduce auxiliary vectors
\begin{equation}
 \bu^{(i)}=(\bI+\bM)^{-1} \bK^i \br,~~
\label{eq:tuij}
\end{equation}
and  $S^{(i,j)}$ defined in \eqref{Sij-2} is then rewritten as
\begin{equation}
S^{(i,j)}= \tc \bK^j(\bI+ \bM)^{-1} \bK^i \br
= \tc \bK^j\bu^{(i)}.\label{Sij-3}
\end{equation}

Using the shift relation \eqref{cond-r-p}, from \eqref{eq:tuij} we have
\[(\bI+\t \bM)\t\bu^{(i)}=\bK^i \t\br=(p\bI-\bK)^{-1}\bK^{i}(p\bI+\bK)\br,\]
i.e.,
\[\bK^{i}(p\bI+\bK)\br=(p\bI-\bK)(\bI+\t \bM)\t\bu^{(i)}.\]
Then, employing the exchange relation
\begin{equation}
(p\bI-\bK)(\bI+\wt{\bM})=(\bI+\bM)(p\bI-\bK)+\br\,\tc \label{eq:M-dyna-2-exc}
\end{equation}
that is indicated by \eqref{eq:M-dyna-2}, one has
\[\bK^{i}(p\bI+\bK)\br=(\bI+\bM)(p\bI-\bK)\t\bu^{(i)}+ \br\, \tc\,\wt{\bu}^{(i)},\]
which further, left multiplied by $(\bI+\bM)^{-1}$, yields the relation
\begin{subequations}
\label{ui-shift}
\begin{eqnarray}
 (p \bI-\bK)\wt{\bu}^{(i)}=p \bu^{(i)}+\bu^{(i+1)}-\wt{S}^{(i,0)}\bu^{(0)}. \label{ui-shift-1}
\end{eqnarray}
Now replacing $\br$ in \eqref{eq:tuij} by using \eqref{cond-r-p}, we have
\[(\bI+ \bM)\bu^{(i)}=\bK^i \br=(p\bI+\bK)^{-1}\bK^{i}(p\bI-\bK)\t\br,\]
i.e.,
\[\bK^{i}(p\bI-\bK)\t\br=(p\bI+\bK)(\bI+ \bM)\bu^{(i)}.\]
In this turn we make use of the exchange relation
\begin{equation*}
(p\bI+\bK)(\bI+{\bM})=(\bI+\t\bM)(p\bI+\bK)-\t\br~\tc
\end{equation*}
that is derived from \eqref{eq:M-dyna-1}, we will finally reach to
\begin{eqnarray}
 (p \bI+\bK){\bu}^{(i)}=p\t \bu^{(i)}-\t \bu^{(i+1)}+{S}^{(i,0)}\t\bu^{(0)}. \label{ui-shift-2}
\end{eqnarray}
Symmetrically, we can derive shift relations for $(q, \h{~})$,
\begin{eqnarray}
&& (q \bI-\bK)\wh{\bu}^{(i)}=q\bu^{(i)}+\bu^{(i+1)}-\wh{S}^{(i,0)}\bu^{(0)},\label{ui-shift-3}\\
&& (q \bI+\bK)\bu^{(i)}=q\wh{\bu}^{(i)} -\wh{\bu}^{(i+1)}+S^{(i,0)}\wh{\bu}^{(0)}.\label{ui-shift-4}
\end{eqnarray}
\end{subequations}
Now, with the shift relation \eqref{ui-shift} in hand, a left-multiplication of ${}^t\bc\bK^j$
immediately yields the recurrence relation \eqref{Sij-shift} for $S^{(i,j)}$.

The symmetric property $S^{(i,j)}=S^{(j,i)}$ plays an important role and we will discuss this property later.
If this property holds, noting that the formal expansion
\begin{equation}
S(a,b)=\sum^{\infty}_{j=0}\sum^{\infty}_{i=0}\frac{(-1)^{i+j}}{a^{i+1}b^{j+1}}S^{(i,j)},
\end{equation}
it is easy to see that $S(a,b)=S(b,a)$,
and similarly, for $V(a)$ defined in \eqref{Va-def}, i.e.,
\begin{equation}
V(a)=1-\tc (\bI+\bM)^{-1}(a\bI+ \bK)^{-1}\br =1-\tc(a \bI+ \bK)^{-1}(\bI+\bM)^{-1}\br,
\label{Va-def-sym}
\end{equation}
the second equality of holds as well.

In order to get the relation \eqref{Sab-2-re}, we introduce an auxiliary vector
\begin{equation}
\bu(a)= (\bI+\bM)^{-1}(a\bI+\bK)^{-1}\br,
\end{equation}
under which we have
\begin{subequations}\label{Sab-Va}
\begin{align}
S(a,b)&=\tc(b\bI+\bK)^{-1}\bu(a),\\
V(a)&=1-\tc\, \bu(a). \label{Va-def-2}
\end{align}
\end{subequations}
Then, using \eqref{cond-r-p} we have
\[
\t\bu(a)=(\bI+\t\bM)^{-1}(a\bI+\bK)^{-1}\t\br=(\bI+\wt{\bM})^{-1}(a\bI+\bK)^{-1}(p\bI-\bK)^{-1}(p\bI+\bK)\br,\]
i.e.,
\[
(p\bI-\bK)(\bI+\wt{\bM})\t\bu(a)=\br+(p-a)(a\bI+\bK)^{-1}\br.
\]
Making use of \eqref{eq:M-dyna-2-exc} and \eqref{Va-def-2} yields
\[
(\bI+\bM)(p\bI-\bK)\t\bu(a)=(p-a)(a\bI+\bK)^{-1}\br+\t V(a)\br.
\]
Then by left-multiplying $(\bI+\bM)^{-1}$ we get
\begin{equation}
(p\bI-\bK)\t\bu(a)=(p-a) \bu(a)+\t V(a)\bu^{(0)}.
\label{ua-t}
\end{equation}

Next, left-multiplying $\tc(b\bI+\bK)^{-1}$ on \eqref{ua-t} we can reach to
\begin{equation}
 1-(p+b)\t{S}(a,b)+(p-a)S(a,b)=\wt{V}(a)V(b),
\end{equation}
i.e., the shift relation \eqref{Sab-2-re-a},
where we have used the expression \eqref{Sab-Va} and the symmetric form
$V(b)=1-\tc (b\bI+\bK)^{-1}(\bI+\bM)^{-1}\br$
which has been given in \eqref{Va-def-sym}.
In a similar way we get relation \eqref{Sab-2-re-b}.
\eqref{Sab-2-re-c} and \eqref{Sab-2-re-d} are derived from \eqref{Sab-2-re-a} and \eqref{Sab-2-re-b}
thanks to the arbitrariness of $a,b$ and the symmetric property $S(a,b)=S(b,a)$.

Now the proof for Theorem \ref{T:1} is completed.
\end{proof}

\subsection{List of equations}
\label{sec:3.2}

With the assumption of the symmetric property \eqref{Sij-sym}, some lattice equations can be derived  as closed forms of
the recurrence relations given in Theorem \ref{T:1} (see \cite{Nijhoff-ABS} for detailed derivation).
In the following we list them out together with their solution formulae.
\begin{itemize}
\item{Lattice potential KdV (lpKdV) equation:
\begin{subequations}
\label{lpKdV}
\begin{eqnarray}
&& (p+q+w-\th{w})(p-q+\h{w}-\t{w})=p^2-q^2,\label{lpKdV-a}\\
&& w=S^{(0,0)}= \tc (\bI+ \bM)^{-1}\br;\label{lpKdV-b}
\end{eqnarray}
\end{subequations}
}
\item{Lattice potential modified KdV (lpmKdV) equation:
\begin{subequations}
\label{lpmKdV}
\begin{eqnarray}
&& p(v\wh{v}-\wt{v}\wh{\wt{v}})=q(v\wt{v}-\wh{v}\wh{\wt{v}}),\label{lpmKdV-a}\\
&& v=1-S^{(0,-1)}= 1-\tc \bK^{-1} (\bI+\bM)^{-1}\br;
\label{lpmKdV-b}
\end{eqnarray}
\end{subequations}
}
\item{Lattice Schwarzian  KdV (lSKdV) equation:
\begin{subequations}
\label{lSKdV}
\begin{eqnarray}
&& \frac{(z-\wt{z})(\wh{z}-\wh{\wt{z}})}{(z-\wh{z})(\wt{z}-\wh{\wt{z}})}=\frac{q^2}{p^2},\label{lSKdV-a}\\
&& z=\tc \bK^{-1}(\bI+ \bM)^{-1}\bK^{-1}\br-z_0-\frac{n}{p}-\frac{m}{q},~~ z_0\in \mathbb{C};
\label{lSKdV-b}
\end{eqnarray}
\end{subequations}
}
\item{NQC equation\cite{NQC-1983}:
\begin{subequations}\label{NQC}
\begin{eqnarray}
&& \frac{1-(p+b)\wh{\wt{S}}(a,b)+(p-a)\wh{S}(a,b)}
{1-(q+b)\wh{\wt{S}}(a,b)+(q-a)\wt{S}(a,b)}
=\frac{1-(q+a)\wh{S}(a,b)+(q-b)S(a,b)}
{1-(p+a)\wt{S}(a,b)+(p-b)S(a,b)},~~~~
\label{NQC-eq}\\
&& S(a,b)=\tc(b\bI+\bK)^{-1}(\bI+\bM)^{-1}(a\bI+\bK)^{-1}\br,~~ a,b\in \mathbb{C}.
\label{Sab-3}
\end{eqnarray}
\end{subequations}
}
\end{itemize}

In the following our task is to solve the starting equation set \eqref{cond} so that we can give explicit solutions of the above lattice equations.

\section{Solutions to the starting matrix equation set \eqref{cond}}

\subsection{Simplification and canonical forms}

In the starting matrix equation set \eqref{cond} $\tc=(c_1,c_2,\cdots, c_N)$ is a known constant vector
and we need to look for the solution triad $\{\bK, \br, \bM\}$.
We suppose $\Ga$ is the canonical form of $\bK$, i.e.,
\begin{equation}
\bK=\bT^{-1} \Ga \,\bT,
\end{equation}
where $\bT$ is the transform matrix. We define
\begin{eqnarray}
\bM_1=\bT \bM \bT^{-1},~~\br_1=\bT \br,~~\tc_1=\tc\, \bT^{-1}.\label{eq:-Mrtc-1}
\end{eqnarray}
It then follows from \eqref{cond} that
\begin{subequations}
\begin{align}
&(p \bI-\Ga)\wt{\br}_1=(p\bI+\Ga) \br_1,\label{cond-1a}\\
&(q \bI-\Ga)\wh{\br}_1=(q\bI+\Ga) \br_1,\label{cond-1b}\\
& \bM_1 \Ga + \Ga\bM_1= \br_1 \,\tc_1. \label{cond-1c}
\end{align}
\label{cond-1}
\end{subequations}
Besides, for the scalar functions $ S^{(i,j)}$, $S(a,b)$ and $V(a)$ we find
\begin{align*}
S^{(i,j)}& = \tc \bK^j(\bI+ \bM)^{-1} \bK^i \br= \tc_1
\Ga^j(\bI+ \bM_1)^{-1}\Ga^i \br_1,\\
S(a,b)& = \tc (b \bI+ \bK)^{-1}(\bI+ \bM)^{-1}(a \bI+ \bK)^{-1}\br= \tc_1
 (b\bI+\Ga)^{-1} (\bI+ \bM_1)^{-1} (a \bI+ \Ga)^{-1} \br_1,\\
V(a)& =1-\tc(\bI+\bM)^{-1}(a \bI+ \bK)^{-1}\br = 1-\tc_1(\bI+\bM_1)^{-1}(a \bI+ \Ga)^{-1}\br_1,
\end{align*}
i.e., $\bK$ and its canonical form $\Ga$ lead to same $ S^{(i,j)}$, same $S(a,b)$ and same $V(a)$.
One more invariant is
\begin{equation}
Z(a,b) = \tc (b \bI+ \bK)^{-2}(\bI+ \bM)^{-1}(a \bI+ \bK)^{-1}\br= \tc_1
 (b\bI+\Ga)^{-2} (\bI+ \bM_1)^{-1} (a \bI+ \Ga)^{-1} \br_1,
\end{equation}
which will be used in later discussion.
Thanks to such invariance, in the following we neglect the subscripts in \eqref{cond-1} and consider the simplified/canonical equation set
\begin{subequations}
\begin{align}
&(p \bI-\Ga)\wt{\br}=(p\bI+\Ga) \br,\label{cond-2a}\\
&(q \bI-\Ga)\wh{\br}=(q\bI+\Ga) \br,\label{cond-2b}\\
& \bM \Ga + \Ga\bM= \br\, \tc, \label{cond-2c}
\end{align}
\label{cond-2}
\end{subequations}
where $\Ga$ is a $N\times N$ matrix in canonical form.
Here we would like to specify the invertible assumption for $\Ga$.
\begin{Proposition}
\label{P:1}
Hereafter we always suppose $\Ga$ satisfies the following  necessary invertible condition:
$s\bI\pm \Ga$ is invertible for $s=0,p,q,a,b$,
and all the eigenvalues $\{k_1,k_2,\cdots, k_N\}$ (some of them can be same) of $\Ga$
satisfy
\begin{equation}
k_i+k_j\neq 0,~~ \forall~ 1\leq i,j\leq N.
\label{kij-3}
\end{equation}
\end{Proposition}

We can then start from the elements in \eqref{cond-2}
and replace Theorem \ref{T:1} as the following.

\begin{Theorem}
\label{T:2}
Suppose that $\Ga$ is a $N\times N$ matrix in canonical form satisfying necessary invertible condition
described in Proposition \ref{P:1}, $\bM, \, \br,\, \tc$ are defined as in Theorem \ref{T:1}.
and they satisfy the equation set \eqref{cond-2}.
By them we define
\begin{subequations}
\label{SSV}
\begin{align}
& S^{(i,j)}=  \tc\, \Ga^j(\bI+ \bM)^{-1}\Ga^i \br,\label{SSV-Sij}\\
& S(a,b)=  \tc (b\bI+\Ga)^{-1} (\bI+ \bM)^{-1} (a \bI+ \Ga)^{-1} \br,\label{SSV-Sab}
\\
&V(a)=1-\tc(\bI+\bM)^{-1}(a \bI+ \Ga)^{-1}\br,
\label{SSV-V}\\
&Z(a,b) = \tc (b \bI+ \Ga)^{-2}(\bI+ \bM)^{-1}(a \bI+ \Ga)^{-1}\br.\label{SSV-Z}
\end{align}
\end{subequations}
Then, $S^{(i,j)}$ obeys the recurrence relation \eqref{Sij-shift}.
If $S^{(i,j)}$ has the symmetric property
\begin{equation}
S^{(i,j)}=S^{(j,i)},
\label{Sij-sym-2}
\end{equation}
then $S(a,b)$ and $V(a)$ also have symmetric property
\begin{subequations}
\label{SV-sym}
\begin{align}
& S(a,b)=  S(b,a),\label{SV-sym-S}\\
&V(a)=1-\tc (\bI+\bM)^{-1}(a\bI+ \bK)^{-1}\br=1-\tc(a \bI+ \bK)^{-1}(\bI+\bM)^{-1}\br,
\label{SV-sym-V}
\end{align}
\end{subequations}
and they satisfy the shift relation \eqref{Sab-2-re}.
\end{Theorem}


\subsection{Solutions to \eqref{cond-2}}\label{sec:4.2}

Noting that in the equation set \eqref{cond-2} $\tc$ is known and $\Ga$ is a matrix in canonical form,
it is possible to give a complete discussion for the solution pair $\{\br,\,\bM\}$ of \eqref{cond-2}
according to the eigenvalue structure of $\Ga$.

\subsubsection{List of notations}\label{sec:4.2.1}

First we need to introduce some notations where usually the subscripts $_D$ and $_J$ correspond to
the cases of $\Ga$ being diagonal and being of Jordan block, respectively.
\begin{subequations}\label{notations}
\begin{align}
& \hbox{plain wave factor:}~~\rho_i=\biggl(\frac{p +k_i}{p -k_i}\biggr)^n\biggl(\frac{q + k_i}{q-k_i}\biggr)^m \rho^0_{i},~~
\mathrm{with~ constants~}p,q,\rho^0_{i},\\
& N\mathrm{\hbox{-}th~order~vector:}~~\br_{\hbox{\tiny{\it D}}}^{\hbox{\tiny{[{\it N}]}}}(\{k_j\}_{1}^{N})=(\rho_1, \rho_2, \cdots, \rho_N)^T,\\
& N\mathrm{\hbox{-}th~order~vector:}~~\br_{\ty{J}}^{\tyb{N}}(k_1)=\Bigl(\rho_1, \frac{\partial_{k_1}\rho_1}{1!},
\cdots, \frac{\partial^{N-1}_{k_1}\rho_1}{(N-1)!}\Bigr)^T,\\
& N\mathrm{\hbox{-}th~order~vector:}~~\bI_{\ty{D}}^{\tyb{N}}=(1,1,1, \cdots, 1)^T,\\
& N\mathrm{\hbox{-}th~order~vector:}~~\bI_{\ty{J}}^{\tyb{N}}=(1,0,0, \cdots, 0)^T,\\
& N\mathrm{\hbox{-}th~order~vector:}~~g^{\tyb{N}}(a)=\Bigl(\frac{1}{a},\frac{-1}{a^2},\frac{1}{a^3}, \cdots,\frac{(-1)^{N-1}}{a^{N}}\Bigr)^T,\\
& N\times N ~\mathrm{matrix:}~~\Ga^{\tyb{N}}_{\ty{D}}(\{k_j\}^{N}_{1})=\mathrm{Diag}(k_1, k_2, \cdots, k_N),\\
& N\times N ~\mathrm{matrix:}~~\Ga^{\tyb{N}}_{\ty{J}}(a)
=\left(\begin{array}{cccccc}
a & 0    & 0   & \cdots & 0   & 0 \\
1   & a  & 0   & \cdots & 0   & 0 \\
0   & 1  & a   & \cdots & 0   & 0 \\
\vdots &\vdots &\vdots &\vdots &\vdots &\vdots \\
0   & 0    & 0   & \cdots & 1   & a
\end{array}\right),\\
& N\times N ~\mathrm{matrix:}~~\bF^{\tyb{N}}_{\ty{D}}(\{k_j\}^{N}_{1})=\mathrm{Diag}(\rho_1, \rho_2, \cdots, \rho_N),\\
& N\times N ~\mathrm{matrix:}~~\bH^{\tyb{N}}_{\ty{D}}(\{c_j\}^{N}_{1})=\mathrm{Diag}(c_1, c_2, \cdots, c_N),\\
& N\times N ~\mathrm{matrix:}~~\bF^{\tyb{N}}_{\ty{J}}(k_1)
=\left(
\begin{array}{ccccc}
\rho_1 & 0 & 0 & \cdots & 0\\
\frac{\partial_{k_1}\rho_1}{1!} & \rho_1 & 0 & \cdots & 0\\
\frac{\partial^{2}_{k_1}\rho_1}{2!} &\frac{\partial_{k_1}\rho_1}{1!} & \rho_1 & \cdots & 0\\
\vdots &\vdots &\vdots & \ddots & \vdots\\
\frac{\partial^{N-1}_{k_1}\rho_1}{(N-1)!} & \frac{\partial^{N-2}_{k_1}\rho_1 }{(N-2)!} & \frac{\partial^{N-3}_{k_1}\rho_1}{(N-3)!} & \cdots & \rho_1
\end{array}
\right),\\
& N\times N ~\mathrm{matrix:}~~\bH^{\tyb{N}}_{\ty{J}}(\{c_j\}^{N}_{1})
=\left(\begin{array}{ccccc}
c_1 & \cdots  & c_{N-2}  & c_{N-1} & c_N\\
c_2 & \cdots & c_{N-1}  & c_N & 0\\
c_3 &\cdots & c_N & 0 & 0\\
\vdots &\vdots & \vdots & \vdots & \vdots\\
c_N & \cdots & 0 & 0 & 0
\end{array}
\right),
\end{align}
\begin{align}
& N\times N ~\mathrm{matrix:}~~\bG^{\tyb{N}}_{\ty{D}}(\{k_j\}^{N}_{1})
=(g_{i,j})_{N\times N},~~~g_{i,j}=\frac{1}{k_i+k_j},\\
& N_1\times N_2 ~\mathrm{matrix:}~~\bG^{\tyb{N$_1$,N$_2$}}_{\ty{DJ}}(\{k_j\}^{N_1}_{1};a)
=(g_{i,j})_{N_1\times N_2},~~~g_{i,j}=-\Bigl(\frac{-1}{k_i+a}\Bigr)^j,\\
& N_1\times N_2 ~\mathrm{matrix:}~~\bG^{\tyb{N$_1$,N$_2$}}_{\ty{JJ}}(a;b)
=(g_{i,j})_{N_1\times N_2},~~~g_{i,j}=\mathrm{C}^{i-1}_{i+j-2}\frac{(-1)^{i+j}}{(a+b)^{i+j-1}},\\
& N\times N ~\mathrm{matrix:}~~\bG^{\tyb{N}}_{\ty{J}}(a)=\bG^{\tyb{N,N}}_{\ty{JJ}}(a;a)
=(g_{i,j})_{N\times N},~~~g_{i,j}=\mathrm{C}^{i-1}_{i+j-2}\frac{(-1)^{i+j}}{(2a)^{i+j-1}},\label{G-J-a}
\end{align}
\end{subequations}
where
\[\mathrm{C}^{i}_{j}=\frac{j!}{i!(j-i)!},~~(j\geq i).\]

The $N$th-order matrix  in the following form
\begin{equation}
\mathcal{A}=\left(\begin{array}{cccccc}
a_0 & 0    & 0   & \cdots & 0   & 0 \\
a_1 & a_0  & 0   & \cdots & 0   & 0 \\
a_2 & a_1  & a_0 & \cdots & 0   & 0 \\
\vdots &\vdots &\cdots &\vdots &\vdots &\vdots \\
a_{N-1} & a_{N-2} & a_{N-3}  & \cdots &  a_1   & a_0
\end{array}\right)_{N\times N}
\label{A}
\end{equation}
with scalar elements $\{a_j\}$ is  a $N$th-order lower triangular Toeplitz matrix.
All such matrices compose a commutative set $\widetilde{G}^{\tyb{N}}$ with respect to matrix multiplication
and the subset
\[G^{\tyb{N}}=\big \{\mathcal{A} \big |~\big. \mathcal{A}\in \widetilde{G}^{\tyb{N}},~|\mathcal{A}|\neq 0 \big\}\]
is an Abelian group.
Such kind of matrices play useful roles in the expression of exact solution for soliton equations.
For more properties of such matrices one can refer to Ref.\cite{Zhang-KdV}.

For the notations we just listed out, it is easy to find the following facts.
\begin{Proposition}
\label{P:2}
~\\
\textrm{(1).}~ Both $\bH^{\tyb{N}}_{\ty{J}}(\{c_j\}^{N}_{1})$ and $\bG^{\tyb{N}}_{\ty{J}}(a)$ are symmetric matrices;\\
\textrm{(2).}~ Both $\Ga^{\tyb{N}}_{\ty{J}}(a)$ and $\bF^{\tyb{N}}_{\ty{J}}(k_1)$ belong to $\widetilde{G}^{\tyb{N}}$
and therefore they commute;\\
\textrm{(3).}~ For $\bH^{\tyb{N}}_{\ty{J}}(\{c_j\}^{N}_{1})$ and $\forall \mathcal{A}\in \widetilde{G}^{\tyb{N}}$, they satisfy
\begin{equation}
\bH^{\tyb{N}}_{\ty{J}}(\{c_j\}^{N}_{1})\mathcal{A}=\mathcal{A}^T \bH^{\tyb{N}}_{\ty{J}}(\{c_j\}^{N}_{1}),
\label{commut-1}
\end{equation}
or, in other words, $\bH^{\tyb{N}}_{\ty{J}}(\{c_j\}^{N}_{1})\mathcal{A}$ is a symmetric matrix;\\
\textrm{(4).}~ When $a\neq 0$, $\Ga^{\tyb{N}}_{\ty{J}}(a)^{-1}$ is given by
\begin{equation}
\Ga^{\tyb{N}}_{\ty{J}}(a)^{-1}=\left(\begin{array}{cccccc}
\frac{1}{a} & 0    & 0   & \cdots & 0   & 0 \\
\frac{-1}{a^2} & \frac{1}{a}  & 0   & \cdots & 0   & 0 \\
\frac{1}{a^3} & \frac{-1}{a^2}  & \frac{1}{a} & \cdots & 0   & 0 \\
\vdots &\vdots &\cdots &\vdots &\vdots &\vdots \\
\frac{(-1)^{N-1}}{a^{N}} & \frac{(-1)^{N-2}}{a^{N-1}} & \frac{(-1)^{N-3}}{a^{N-2}}  & \cdots & \frac{-1}{a^2} & \frac{1}{a}
\end{array}\right).
\label{A-inv}
\end{equation}
\end{Proposition}

\subsubsection{Basic cases}\label{sec:4.2.2}

Let us discuss solutions for the equation set \eqref{cond-2} when $\Ga$ takes two basic forms.
Notations can be referred to Sec.\ref{sec:4.2.1}.

\vskip 5pt
\noindent
\textit{Case 1.~} $\Ga=\Ga^{\tyb{N}}_{\ty{D}}(\{k_j\}^{N}_{1})$.

This case has been solved in Ref.\cite{Nijhoff-ABS} as well as in Sec.\ref{sec:2}. Solutions to \eqref{cond-2} are
\begin{subequations}
\begin{align}
\br & =\br_{\hbox{\tiny{\it D}}}^{\hbox{\tiny{[{\it N}]}}}(\{k_j\}_{1}^{N}),\\
\bM & =\bF^{\tyb{N}}_{\ty{D}}(\{k_j\}^{N}_{1})\bG^{\tyb{N}}_{\ty{D}}(\{k_j\}^{N}_{1})\bH^{\tyb{N}}_{\ty{D}}(\{c_j\}^{N}_{1})
=\Bigl(\frac{\rho_i c_j}{k_i+k_j}\Bigr)_{N\times N}.
\end{align}
\label{sol-diag}
\end{subequations}
In this case one gets soliton solutions \cite{Nijhoff-ABS}.

\vskip 5pt
\noindent
\textit{Case 2.~} $\Ga=\Ga^{\tyb{N}}_{\ty{J}}(k_1)$.

This is also called Jordan block case. Motivated by the Jordan block case of many other soliton equations\cite{Zhang-KdV,Hietarinta-Bou},
it is not difficult to find a solution for the equations \eqref{cond-2a} and  \eqref{cond-2b},
\begin{equation}
\br=\br_{\ty{J}}^{\tyb{N}}(k_1).
\end{equation}
To find a solution $\bM$ to \eqref{cond-2c}, we first rewrite
\begin{equation}
\br_{\ty{J}}^{\tyb{N}}(k_1)=\bF\cdot\bI^{\tyb{N}}_{\ty{J}},~~~
\tc={\bI^{\tyb{N}}_{\ty{J}}}^T\cdot \bH,~~~
\bM=\bF  \bG  \bH,
\end{equation}
where $\bF=\bF^{\tyb{N}}_{\ty{J}}(k_1)$, $\bH=\bH^{\tyb{N}}_{\ty{J}}(\{c_j\}^{N}_{1})$ and
$\bG$ is an unknown matrix.
It then follows from \eqref{cond-2c} that
\begin{equation*}
\bF \bG\bH \Ga  +\Ga  \bF \bG\bH
=\bF\cdot\bI^{\tyb{N}}_{\ty{J}}\cdot {\bI^{\tyb{N}}_{\ty{J}}}^T\cdot \bH,
\end{equation*}
where $\Ga=\Ga^{\tyb{N}}_{\ty{J}}(k_1)$.
Noting that the item (2) and (3) in Proposition \ref{P:2}, we then have
\begin{equation*}
\bF \bG {\Ga}^T \bH + \bF \Ga \bG\bH
=\bF\cdot\bI^{\tyb{N}}_{\ty{J}}\cdot {\bI^{\tyb{N}}_{\ty{J}}}^T\cdot \bH,
\end{equation*}
which further reduces to
\begin{equation}
\bG {\Ga}^T  +  \Ga \bG
=\bI^{\tyb{N}}_{\ty{J}}\cdot {\bI^{\tyb{N}}_{\ty{J}}}^T.
\label{G-J-eq}
\end{equation}
Thus we need to solve the above equation. To do that, we write
\begin{equation}
\bG=(g_1,g_2,\cdots,g_N),
\label{G-J-c}
\end{equation}
where $\{g_j\}$ are column vectors of $\bG$.
Then from \eqref{G-J-eq} we find
\begin{align*}
& (k_1+ \Ga) g_1= \bI^{\tyb{N}}_{\ty{J}},\\
& (k_1+ \Ga) g_{j+1}+g_{j}=0,~~ ~(j=1,2,\cdots, N-1).
\end{align*}
This indicates solutions
\begin{subequations}
\begin{align}
& g_1={\Ga^{\tyb{N}}_{\ty{J}}(2k_1)}^{-1}\cdot {\bI^{\tyb{N}}_{\ty{J}}}=g^{\tyb{N}}(2k_1),\\
&  g_{j+1}=\frac{\partial_a^j g^{\tyb{N}}(a)|_{a=2k_1}}{j!},~~ ~(j=1,2,\cdots, N-1),
\end{align}
\label{g-C2}
\end{subequations}
where we have made use of the item (4) in Proposition \ref{P:2}.
Substituting \eqref{g-C2} into \eqref{G-J-c} we can find
\[\bG=\bG^{\tyb{N}}_{\ty{J}}(k_1)\]
with explicit expression \eqref{G-J-a}.

Thus, in this case, the solution pair to \eqref{cond-2} is given by
\begin{subequations}\label{rM-J}
\begin{align}
& \br=\br_{\ty{J}}^{\tyb{N}}(k_1),\\
& \bM= \bF^{\tyb{N}}_{\ty{J}}(k_1)\cdot \bG^{\tyb{N}}_{\ty{J}}(k_1)\cdot \bH^{\tyb{N}}_{\ty{J}}(\{c_j\}^{N}_{1}).
\end{align}
\end{subequations}
Besides, using the commutative property of lower triangular Toeplitz matrices in the  set $\t G^{\tyb{N}}$ we have the following result.
\begin{Proposition}\label{P:3}
For an arbitrary constant matrix $\mathcal{A}\in \t G^{\tyb{N}}$, $\mathcal{A}\br$ and $\mathcal{A}\bM$
solve the equation set \eqref{cond-2},
where $\br$  and $\bM$ are defined in \eqref{rM-J}.
\end{Proposition}

\subsubsection{Generic case}\label{sec:4.2.3}

To investigate solutions to the equation set \eqref{cond-2} with a generic $\Ga$,
we need to consider some elementary equations.

\begin{lem}
\label{L:4-1}
Consider the following three elementary matrix equations:
\begin{equation}
\bG \,\Ga^{\tyb{N$_1$}}_{\ty{D}}(\{k_j\}^{N_1}_{1})+ \Ga^{\tyb{N$_1$}}_{\ty{D}}(\{k_j\}^{N_1}_{1})\bG
= \bI_{\ty{D}}^{\tyb{N$_1$}}\cdot {\bI_{\ty{D}}^{\tyb{N$_1$}}}^T,
\label{ele-eq-1}
\end{equation}
\begin{equation}
\bG \,{\Ga^{\tyb{N$_2$}}_{\ty{J}}(b)}^T+ \Ga^{\tyb{N$_1$}}_{\ty{D}}(\{k_j\}^{N_1}_{1})\,\bG
= \bI_{\ty{D}}^{\tyb{N$_1$}}\cdot {\bI_{\ty{J}}^{\tyb{N$_2$}}}^T,
\label{ele-eq-2}
\end{equation}
\begin{equation}
\bG\, {\Ga^{\tyb{N$_2$}}_{\ty{J}}(b)}^T+ \Ga^{\tyb{N$_1$}}_{\ty{J}}(a)\,\bG
= \bI_{\ty{J}}^{\tyb{N$_1$}}\cdot {\bI_{\ty{J}}^{\tyb{N$_2$}}}^T,
\label{ele-eq-3}
\end{equation}
where the unknown matrix $\bG$ in the above three equations is of respectively $N_1\times N_1$, $N_1\times N_2$ and $N_1\times N_2$.
For solutions  we have
$\bG=\bG^{\tyb{N$_1$}}_{\ty{D}}(\{k_j\}^{N_1}_{1})$ solves \eqref{ele-eq-1},
$\bG=\bG^{\tyb{N$_1$,N$_2$}}_{\ty{DJ}}(\{k_j\}^{N_1}_{1};b)$ solves \eqref{ele-eq-2},
and $\bG=\bG^{\tyb{N$_1$,N$_2$}}_{\ty{JJ}}(a;b)$ solves \eqref{ele-eq-3}.
\end{lem}

\begin{proof}
$\bG=\bG^{\tyb{N$_1$}}_{\ty{D}}(\{k_j\}^{N_1}_{1})$ satisfying \eqref{ele-eq-1} can be verified directly.

Let us consider  equation \eqref{ele-eq-2}.
We suppose
\begin{equation}
\bG=(g_1,g_2,\cdots,g_{N_2})
\label{G-DJ-c}
\end{equation}
with column vectors $\{g_j\}$.
Then from \eqref{ele-eq-2} we have
\begin{align*}
& (b+\Ga^{\tyb{N$_1$}}_{\ty{D}}(\{k_j\}^{N}_{1})) g_1= \bI^{\tyb{N$_1$}}_{\ty{D}},\\
& (b+\Ga^{\tyb{N$_1$}}_{\ty{D}}(\{k_j\}^{N}_{1})) g_{j+1}+g_{j}=0,~~ ~(j=1,2,\cdots, N_2-1),
\end{align*}
which further yields
\[  g_{j}=(-1)^{j-1} \Bigl( \frac{1}{(k_1+b)^j},\frac{1}{(k_2+b)^j},\cdots,\frac{1}{(k_{N_1}+b)^j}\Bigr)^T,~~ ~(j=1,2,\cdots, N_2).
\]
The matrix \eqref{G-DJ-c} with above $\{g_j\}$ is just $\bG^{\tyb{N$_1$,N$_2$}}_{\ty{DJ}}(\{k_j\}^{N_1}_{1};b)$
and it solves \eqref{ele-eq-2}.

The solving procedure for \eqref{ele-eq-3} is quite similar to the Case 2 in Sec.\ref{sec:4.2.2}.
We rewrite $\bG$ as \eqref{G-DJ-c} and insert it into \eqref{ele-eq-3}
and then we find
\begin{align*}
& (b+ \Ga^{\tyb{N$_1$}}_{\ty{J}}(a)) g_1= \bI^{\tyb{N$_1$}}_{\ty{J}},\\
& (b+ \Ga^{\tyb{N$_1$}}_{\ty{J}}(a)) g_{j+1}+g_{j}=0,~~ ~(j=1,2,\cdots, N_2-1),
\end{align*}
which yields
\begin{subequations}
\begin{align}
& g_1=g^{\tyb{N$_1$}}(a+b),\\
&  g_{j+1}=\frac{\partial_{\beta}^j g^{\tyb{N$_1$}}(\beta)|_{\beta=a+b}}{j!},~~ ~(j=1,2,\cdots, N_2-1).
\end{align}
\end{subequations}
These column vectors compose a matrix $\bG=\bG^{\tyb{N$_1$,N$_2$}}_{\ty{JJ}}(a;b)$ that solves \eqref{ele-eq-3}.
\end{proof}

Now we can investigate the case with a generic $\Ga$:
\begin{equation}
\Ga=\mathrm{Diag}\bigl(\Ga^{\tyb{N$_1$}}_{\ty{D}}(\{k_j\}^{N_1}_{1}),
\Ga^{\tyb{N$_2$}}_{\ty{J}}(k_{N_1+1}),\Ga^{\tyb{N$_3$}}_{\ty{J}}(k_{N_1+2}),\cdots,
\Ga^{\tyb{N$_s$}}_{\ty{J}}(k_{N_1+(s-1)})\bigr).
\label{Ga-gen}
\end{equation}
For convenience we define column vectors
\begin{equation}
\br=\left(
\begin{array}{c}
\br_{\ty{D}}^{\tyb{N$_1$}}(\{k_j\}_{1}^{N_1})\\
\br_{\ty{J}}^{\tyb{N$_2$}}(k_{N_1+1})\\
\br_{\ty{J}}^{\tyb{N$_3$}}(k_{N_1+2})\\
\vdots\\
\br_{\ty{J}}^{\tyb{N$_s$}}(k_{N_1+(s-1)})
\end{array}
\right),~~~
\bI_{\ty{G}}=\left(
\begin{array}{c}
\bI_{\ty{D}}^{\tyb{N$_1$}}\\
\bI_{\ty{J}}^{\tyb{N$_2$}}\\
\bI_{\ty{J}}^{\tyb{N$_3$}}\\
\vdots\\
\bI_{\ty{J}}^{\tyb{N$_s$}}
\end{array}
\right).
\label{RIG}
\end{equation}
Obviously, such $\br$ solves \eqref{cond-2a} and \eqref{cond-2b} with $\Ga$ defined in \eqref{Ga-gen} and
\begin{equation}
\tc=(c_1,c_2,\cdots,c_{N_1},c_{N_1+1},\cdots, c_{N_1+N_2+\cdots+N_s}).
\label{ct}
\end{equation}

To solve \eqref{cond-2c}, we still suppose
\begin{subequations}
\label{M-g}
\begin{equation}
\bM=\bF\bG \bH,
\end{equation}
where
\begin{align}
&\bF=\mathrm{Diag}\bigl(
\bF^{\tyb{N$_1$}}_{\ty{D}}(\{k_j\}^{N_1}_{1}),
\bF^{\tyb{N$_2$}}_{\ty{J}}(k_{N_1+1}),\bF^{\tyb{N$_3$}}_{\ty{J}}(k_{N_1+2}),\cdots,
\bF^{\tyb{N$_s$}}_{\ty{J}}(k_{N_1+(s-1)})
\bigr),\\
&\bH=\mathrm{Diag}\bigl(
\bH^{\tyb{N$_1$}}_{\ty{D}}(\{c_j\}^{N_1}_{1}),
\bH^{\tyb{N$_2$}}_{\ty{J}}(\{c_j\}^{N_1+N_2}_{N_1+1}),
\cdots,
\bH^{\tyb{N$_s$}}_{\ty{J}}(\{c_j\}^{N_1+N_2+\cdots+N_s}_{N_1+N_2+\cdots+N_{s-1}+1})\bigr),
\end{align}
$\bG$ is a symmetric matrix with block structure
\begin{equation}
\bG=\bG^T=(\bG_{i,j})_{s\times s}
\end{equation}
\end{subequations}
and each $\bG_{i,j}$ is a $N_i\times N_j$ matrix.
Noting that
\[\bF\Ga=\Ga\bF,~~ \bH\Ga=\Ga^T\bH,~~ \br=\bF \bI_{\ty{G}},~~ \tc=\bI_{\ty{G}}^T \bH,\]
from \eqref{cond-2c} we reach to
\begin{equation}
\bG\Ga^T +\Ga\bG = \bI_{\ty{G}}\cdot \bI_{\ty{G}}^T.
\end{equation}
In terms of block structure it reads
\begin{subequations}\label{bs-eq}
\begin{align}
& \bG_{1,1} \,\Ga^{\tyb{N$_1$}}_{\ty{D}}(\{k_j\}^{N_1}_{1})+ \Ga^{\tyb{N$_1$}}_{\ty{D}}(\{k_j\}^{N}_{1})\bG_{1,1}
= \bI_{\ty{D}}^{\tyb{N$_1$}}\cdot {\bI_{\ty{D}}^{\tyb{N$_1$}}}^T,
\label{bs-eq-1}\\
& \bG_{1,j} \,{\Ga^{\tyb{N$_j$}}_{\ty{J}}(k_{N_{j-1}+1})}^T+ \Ga^{\tyb{N$_1$}}_{\ty{D}}(\{k_j\}^{N_1}_{1})\,\bG_{1,j}
= \bI_{\ty{D}}^{\tyb{N$_1$}}\cdot {\bI_{\ty{J}}^{\tyb{N$_j$}}}^T,~~~(1<j\leq s),
\label{bs-eq-2}
\\
& \bG_{i,j}\, {\Ga^{\tyb{N$_j$}}_{\ty{J}}(k_{N_{j-1}+1})}^T+ \Ga^{\tyb{N$_i$}}_{\ty{J}}(k_{N_{i-1}+1})\,\bG_{i,j}
= \bI_{\ty{J}}^{\tyb{N$_i$}}\cdot {\bI_{\ty{J}}^{\tyb{N$_j$}}}^T,~~~(1<i\leq j\leq s).
\label{bs-eq-3}
\end{align}
\end{subequations}
Equations for $\bG_{j,1}$ and $\bG_{j,i}$ are the  transpositions of \eqref{bs-eq-2} and \eqref{bs-eq-3}, respectively,
due to $\bG=\bG^T$.
These equations in \eqref{bs-eq} are just the type of the three elementary matrix equations \eqref{ele-eq-1}, \eqref{ele-eq-2}
and \eqref{ele-eq-3}. Thus based on Lemma \ref{L:4-1} we have solutions
\begin{subequations}\label{bs-sol}
\begin{align}
& \bG_{1,1}=\bG^{\tyb{N$_1$}}_{\ty{D}}(\{k_j\}^{N_1}_{1}), \label{bs-sol-1}\\
& \bG_{1,j}=\bG^{\tyb{N$_1$,N$_j$}}_{\ty{DJ}}(\{k_j\}^{N_1}_{1};k_{N_{j-1}+1}),~~~(1<j\leq s), \label{bs-sol-2}\\
& \bG_{i,j}=\bG^{\tyb{N$_i$,N$_j$}}_{\ty{JJ}}(k_{N_{i-1}+1},k_{N_{j-1}+1}),~~~(1<i\leq j\leq s), \label{bs-sol-3}
\end{align}
\end{subequations}
and then from \eqref{M-g} we get $\bM$.

Let us conclude this part as the following.
\begin{Theorem}
\label{T:3}
For the equation set \eqref{cond-2}, i.e.,
\begin{subequations}
\begin{align}
&(p \bI-\Ga)\wt{\br}=(p\bI+\Ga) \br,\label{cond-3a}\\
&(q \bI-\Ga)\wh{\br}=(q\bI+\Ga) \br,\label{cond-3b}\\
& \bM \Ga + \Ga\bM= \br\, \tc, \label{cond-3c}
\end{align}
\label{cond-3}
\end{subequations}
with generic
\begin{equation}
\Ga=\mathrm{Diag}\bigl(\Ga^{\tyb{N$_1$}}_{\ty{D}}(\{k_j\}^{N_1}_{1}),
\Ga^{\tyb{N$_2$}}_{\ty{J}}(k_{N_1+1}),\Ga^{\tyb{N$_3$}}_{\ty{J}}(k_{N_1+2}),\cdots,
\Ga^{\tyb{N$_s$}}_{\ty{J}}(k_{N_1+(s-1)})\bigr) 
\label{Ga-gen-T}
\end{equation}
and
\begin{equation}
\tc=(c_1,c_2,\cdots,c_{N_1},c_{N_1+1},\cdots, c_{N_1+N_2+\cdots+N_s}),
\label{ct-T}
\end{equation}
we have solutions
\begin{subequations}
\label{r-M-g}
\begin{equation}
\br=\left(
\begin{array}{c}
\br_{\ty{D}}^{\tyb{N$_1$}}(\{k_j\}_{1}^{N_1})\\
\br_{\ty{J}}^{\tyb{N$_2$}}(k_{N_1+1})\\
\br_{\ty{J}}^{\tyb{N$_3$}}(k_{N_1+2})\\
\vdots\\
\br_{\ty{J}}^{\tyb{N$_s$}}(k_{N_1+(s-1)})
\end{array}
\right),~~~
\bM=\bF\bG \bH,
\end{equation}
where
\begin{align}
&\bF=\mathrm{Diag}\bigl(
\bF^{\tyb{N$_1$}}_{\ty{D}}(\{k_j\}^{N_1}_{1}),
\bF^{\tyb{N$_2$}}_{\ty{J}}(k_{N_1+1}),\bF^{\tyb{N$_3$}}_{\ty{J}}(k_{N_1+2}),\cdots,
\bF^{\tyb{N$_s$}}_{\ty{J}}(k_{N_1+(s-1)})
\bigr),\label{r-M-g-F}\\
&\bH=\mathrm{Diag}\bigl(
\bH^{\tyb{N$_1$}}_{\ty{D}}(\{c_j\}^{N_1}_{1}),
\bH^{\tyb{N$_2$}}_{\ty{J}}(\{c_j\}^{N_1+N_2}_{N_1+1}),
\cdots,
\bH^{\tyb{N$_s$}}_{\ty{J}}(\{c_j\}^{N_1+N_2+\cdots+N_s}_{N_1+N_2+\cdots+N_{s-1}+1})\bigr),\label{r-M-g-H}
\end{align}
and $\bG$ is a symmetric matrix with block structure
\begin{equation}
\bG=\bG^T=(\bG_{i,j})_{s\times s}
\label{r-M-g-G1}
\end{equation}
with
\begin{equation}
\begin{array}{ll}
\bG_{1,1}=\bG^{\tyb{N$_1$}}_{\ty{D}}(\{k_j\}^{N_1}_{1}),&~\\
\bG_{1,j}=\bG_{j,1}^T=\bG^{\tyb{N$_1$,N$_j$}}_{\ty{DJ}}(\{k_j\}^{N_1}_{1};k_{N_{j-1}+1}),&~~(1<j\leq s), \\
\bG_{i,j}=\bG_{j,i}^T=\bG^{\tyb{N$_i$,N$_j$}}_{\ty{JJ}}(k_{N_{i-1}+1},k_{N_{j-1}+1}),&~~(1<i\leq j\leq s).
\end{array}
\label{r-M-g-G2}
\end{equation}
\end{subequations}
Besides, for $\Ga$, $\tc$, $\br$ and $\bM$ above mentioned,
the pair
\begin{equation}
\{\mathcal{A}\br, \, \mathcal{A}\bM\}
\label{ArM}
\end{equation}
also solve equation set \eqref{cond-2} with same $\Ga$ and $\tc$,
where
\[\mathcal{A}=Diag(\bI_{N_1},\mathcal{A}_2,\mathcal{A}_3,\cdots,\mathcal{A}_s),\]
in which $\bI_{N_1}$ is the $N_1$-th order unit matrix and $\mathcal{A}_j$ is
a $N_j$-th order constant lower triangular Toeplitz matrix.
In fact, $\mathcal{A}\Ga=\Ga\mathcal{A}$.
\end{Theorem}

Here, $\{\mathcal{A}\br, \, \mathcal{A}\bM\}$ gives a general solution to \eqref{cond-2} with a generic $\Ga$.
It is then quite easy to write out solutions for special cases.
For example, when
\begin{equation}
\Ga=
\left(\begin{array}{cc}
\Ga^{\tyb{N$_1$}}_{\ty{D}}(\{k_j\}^{N_1}_{1})& \mathbf{0}\\
\mathbf{0} & \Ga^{\tyb{N$_2$}}_{\ty{J}}(k_{N_1+1})
\end{array}\right),
~~\tc=(c_1,c_2,\cdots,c_{N_1+N_2}),
\label{Ga-DJ}
\end{equation}
a solution pair to \eqref{cond-3} are
\begin{subequations}
\label{r-M-1}
\begin{equation}
\br=\left(
\begin{array}{c}
\br_{\ty{D}}^{\tyb{N$_1$}}(\{k_j\}_{1}^{N_1})\\
\br_{\ty{J}}^{\tyb{N$_2$}}(k_{N_1+1})
\end{array}
\right),
\end{equation}
and
\begin{equation}
\bM=\bF\bG \bH=
\bF\left(\begin{array}{cc}
\bG^{\tyb{N$_1$}}_{\ty{D}}(\{k_j\}^{N_1}_{1})& \bG^{\tyb{N$_1$,N$_2$}}_{\ty{DJ}}(\{k_j\}^{N_1}_{1};k_{N_{1}+1})\\
{\bG^{\tyb{N$_1$,N$_2$}}_{\ty{DJ}}(\{k_j\}^{N_1}_{1};k_{N_{1}+1})}^T & \bG^{\tyb{N$_2$}}_{\ty{J}}(k_{N_{1}+1})
\end{array}\right)\bH,
\end{equation}
where
\begin{equation}
\bF=\left(\begin{array}{cc}
\bF^{\tyb{N$_1$}}_{\ty{D}}(\{k_j\}^{N_1}_{1})& \mathbf{0}\\
\mathbf{0} & \bF^{\tyb{N$_2$}}_{\ty{J}}(k_{N_1+1})
\end{array}\right),~~
\bH=\left(\begin{array}{cc}
\bH^{\tyb{N$_1$}}_{\ty{D}}(\{c_j\}^{N_1}_{1})& \mathbf{0}\\
\mathbf{0} & \bH^{\tyb{N$_2$}}_{\ty{J}}(\{c_j\}^{N_1+N_2}_{N_1+1})
\end{array}\right).
\end{equation}
\end{subequations}
A second example, when
\begin{equation}
\Ga=
\left(\begin{array}{cc}
\Ga^{\tyb{N$_1$}}_{\ty{J}}(k_1)& \mathbf{0}\\
\mathbf{0} & \Ga^{\tyb{N$_2$}}_{\ty{J}}(k_{2})
\end{array}\right),
~~\tc=(c_1,c_2,\cdots,c_{N_1+N_2}),
\label{Ga-JJ}
\end{equation}
a solution pair to \eqref{cond-3} are
\begin{subequations}
\label{r-M-2}
\begin{equation}
\br=\left(
\begin{array}{c}
\br_{\ty{J}}^{\tyb{N$_1$}}(k_1)\\
\br_{\ty{J}}^{\tyb{N$_2$}}(k_2)
\end{array}
\right),
\end{equation}
and
\begin{equation}
\bM=\bF\bG \bH=
\bF\left(\begin{array}{cc}
\bG^{\tyb{N$_1$}}_{\ty{J}}(k_1)& \bG^{\tyb{N$_1$,N$_2$}}_{\ty{JJ}}(k_1;k_2)\\
{\bG^{\tyb{N$_1$,N$_2$}}_{\ty{JJ}}(k_1;k_2)}^T & \bG^{\tyb{N$_2$}}_{\ty{J}}(k_2)
\end{array}\right)\bH,
\end{equation}
where
\begin{equation}
\bF=\left(\begin{array}{cc}
\bF^{\tyb{N$_1$}}_{\ty{J}}(k_1)& \mathbf{0}\\
\mathbf{0} & \bF^{\tyb{N$_2$}}_{\ty{J}}(k_2)
\end{array}\right),~~
\bH=\left(\begin{array}{cc}
\bH^{\tyb{N$_1$}}_{\ty{J}}(\{c_j\}^{N_1}_{1})& \mathbf{0}\\
\mathbf{0} & \bH^{\tyb{N$_2$}}_{\ty{J}}(\{c_j\}^{N_1+N_2}_{N_1+1})
\end{array}\right).
\end{equation}
\end{subequations}

\subsection{Symmetric property $S^{(i,j)}=S^{(j,i)}$}
\label{sec:4.3}

The symmetric property $S^{(i,j)}=S^{(j,i)}$ directly leads to the
symmetric property $S(a,b)=S(b,a)$ and symmetric definition of $V(a)$.
Therefore this property is crucial for getting lattice equations from the recurrence
relations \eqref{Sij-shift} and \eqref{Sab-2-re}.
In this subsection we will see such a property $S^{(i,j)}=S^{(j,i)}$ holds for a generic $\Ga$ in canonical form.
For the notations we use below please refer to Sec.\ref{sec:4.2.1}.

\subsubsection{Two basic cases}

We start from two basic cases.

\vskip 5pt
\noindent
\textit{Case 1.}~ $\Ga=\Ga^{\tyb{N}}_{\ty{D}}(\{k_j\}^{N}_{1})$.

Solutions to \eqref{cond-2} of this case are given in \eqref{sol-diag}, i.e.,
\begin{equation}
\br  =\br_{\hbox{\tiny{\it D}}}^{\hbox{\tiny{[{\it N}]}}}(\{k_j\}_{1}^{N}),~~
\bM  =\bF \bG \bH,
\label{sol-diag-2}
\end{equation}
where
\begin{equation}
\bF=\bF^{\tyb{N}}_{\ty{D}}(\{k_j\}^{N}_{1}),~~
\bH=\bH^{\tyb{N}}_{\ty{D}}(\{c_j\}^{N}_{1}),~~
\bG=\bG^{\tyb{N}}_{\ty{D}}(\{k_j\}^{N}_{1}).
\end{equation}
Also noting that
\[\br=\bF\cdot\bI^{\tyb{N}}_{\ty{D}},~~~
\tc={\bI^{\tyb{N}}_{\ty{D}}}^T\cdot \bH,\]
from the definition \eqref{SSV-Sij} we have
\begin{align*}
 S^{(i,j)} & =  \tc\, \Ga^j(\bI+ \bM)^{-1}\Ga^i \br\\
 & ={\bI^{\tyb{N}}_{\ty{D}}}^T\cdot \bH \Ga^j(\bI+ \bF\bG\bH)^{-1}\Ga^i\bF\cdot\bI^{\tyb{N}}_{\ty{D}}.
\end{align*}
Since $\Ga$, $\bF$ and $\bH$ are diagonal we can freely commute them and then have
\begin{equation}
S^{(i,j)}  ={\bI^{\tyb{N}}_{\ty{D}}}^T\cdot \Ga^j((\bH\bF)^{-1} + \bG)^{-1}\Ga^i\cdot\bI^{\tyb{N}}_{\ty{D}}.
\label{Sij-dig}
\end{equation}
Noting that $S^{(i,j)}$ is a scalar function and $\bG=\bG^T$, taking transposition for \eqref{Sij-dig} we immediately reach to
\begin{equation}
S^{(i,j)}= {S^{(i,j)}}^T  ={\bI^{\tyb{N}}_{\ty{D}}}^T\cdot \Ga^i((\bH\bF)^{-1} + \bG)^{-1}\Ga^j\cdot\bI^{\tyb{N}}_{\ty{D}}
=S^{(j,i)},
\end{equation}
i.e., the symmetric property $S^{(i,j)}=S^{(j,i)}$.

\vskip 5pt
\noindent
\textit{Case 2.~} $\Ga=\Ga^{\tyb{N}}_{\ty{J}}(k_1)$.

In this case,  the equation set \eqref{cond-2} admits the following expression:
\begin{subequations}
\label{sol-J}
\begin{equation}
\br=\br_{\ty{J}}^{\tyb{N}}(k_1)=\bF\cdot\bI^{\tyb{N}}_{\ty{J}},~~~
\tc={\bI^{\tyb{N}}_{\ty{J}}}^T\cdot \bH,~~~
\bM=\bF  \bG  \bH,
\end{equation}
where
\begin{equation}
\bF=\bF^{\tyb{N}}_{\ty{J}}(k_1), ~~\bH=\bH^{\tyb{N}}_{\ty{J}}(\{c_j\}^{N}_{1}),~~
\bG=\bG^{\tyb{N}}_{\ty{J}}(k_1).
\end{equation}
\end{subequations}
Recalling Proposition \ref{P:2} we know that
\begin{equation}
 \bG=\bG^T,~~\bH=\bH^T,~~
 \bF\Ga=\Ga\bF,~~ \bH\Ga=\Ga^T\bH,~~ \bH \bF=\bF^T \bH,
\label{commut-2}
\end{equation}
in which the last equality indicates $(\bH \bF)=(\bH \bF)^T$.
With these in hand, inserting \eqref{sol-J} into the definition \eqref{SSV-Sij} we have
\begin{equation}
\begin{array}{rl}
 S^{(i,j)} & =  \tc\, \Ga^j(\bI+ \bM)^{-1}\Ga^i \br\\
 & ={\bI^{\tyb{N}}_{\ty{J}}}^T\cdot \bH \Ga^j(\bI+ \bF\bG\bH)^{-1}\Ga^i\bF\cdot\bI^{\tyb{N}}_{\ty{J}}\\
 & ={\bI^{\tyb{N}}_{\ty{J}}}^T\cdot (\Ga^T)^j \bH (\bI+ \bF\bG\bH)^{-1}\bF\Ga^i\cdot\bI^{\tyb{N}}_{\ty{J}}\\
 & = {\bI^{\tyb{N}}_{\ty{J}}}^T\cdot (\Ga^T)^j((\bH\bF)^{-1} + \bG)^{-1}\Ga^i\cdot\bI^{\tyb{N}}_{\ty{J}}.
\end{array}
\label{Sij-J}
\end{equation}
Then, taking transposition we immediately reach to the symmetric property $S^{(i,j)}=S^{(j,i)}$.

\subsubsection{Generic case}

Based on the above two basic cases, now for the generic case the discussion becomes easy.

Go back to Sec. \ref{sec:4.2.3}. For the equation set \eqref{cond-2} with generic
\begin{equation}
\Ga=\mathrm{Diag}\bigl(\Ga^{\tyb{N$_1$}}_{\ty{D}}(\{k_j\}^{N_1}_{1}),
\Ga^{\tyb{N$_2$}}_{\ty{J}}(k_{N_1+1}),\Ga^{\tyb{N$_3$}}_{\ty{J}}(k_{N_1+2}),\cdots,
\Ga^{\tyb{N$_s$}}_{\ty{J}}(k_{N_1+(s-1)})\bigr),
\label{Ga-gen-T-2}
\end{equation}
one has
\begin{equation}
\br=\bF\cdot\bI_{\ty{G}},~~~
\tc={\bI_{\ty{G}}}^T\cdot \bH,~~~
\bM=\bF  \bG  \bH,
\end{equation}
where $\bI_{\ty{G}}$ is defined in \eqref{RIG}, $\bF$, $\bH$ and $\bG$ are given in \eqref{r-M-g}.
A fact is that here $\Ga$, $\bF$, $\bH$ and $\bG$ satisfy exactly  the same relation \eqref{commut-2},
which means in this case $S^{(i,j)}$ will have a same expression as \eqref{Sij-J},
and that means $S^{(i,j)}=S^{(j,i)}$.

\begin{Theorem}\label{T:4}
For the equation set \eqref{cond-2} with $\Ga$ in canonical form,
$ S^{(i,j)}  =  \tc\, \Ga^j(\bI+ \bM)^{-1}\Ga^i \br$ has the
symmetric property $S^{(i,j)}=S^{(j,i)}$.
\end{Theorem}

\section{Applications}

\subsection{Solutions to some lattice equations}\label{sec:5.1}

Now, for the lattice equations we listed out in Sec.\ref{sec:3.2},
their solutions can be given as follows, respectively:
for the lpKdV equation \eqref{lpKdV-a},
\begin{equation}
w=S^{(0,0)}= \tc (\bI+ \bM)^{-1}\br;\label{lpKdV-c}
\end{equation}
for the lpmKdV equation \eqref{lpmKdV-a},
\begin{equation}
v=1-S^{(0,-1)}= 1-\tc \Ga^{-1} (\bI+\bM)^{-1}\br;
\label{lpmKdV-c}
\end{equation}
for the lSKdV equation \eqref{lSKdV-a},
\begin{equation}
z=\tc \Ga^{-1}(\bI+ \bM)^{-1}\Ga^{-1}\br-z_0-\frac{n}{p}-\frac{m}{q},~~ z_0\in \mathbb{C};
\label{lSKdV-c}
\end{equation}
for the NQC equation \eqref{NQC-eq}
\begin{equation}
S(a,b)=\tc(b\bI+\Ga)^{-1}(\bI+\bM)^{-1}(a\bI+\Ga)^{-1}\br,~~ a,b\in \mathbb{C},
\label{Sab-nqc}
\end{equation}
where $\br$, $\tc$, $\Ga$ and $\bM$ are described as in Theorem \ref{T:3}.

\subsection{Solutions to the lattice equations in  ABS List}\label{sec:5.2}

\subsubsection{Re-parametrisation}\label{sec:5.2.1}

Ref.\cite{Nijhoff-ABS} builds the relation between NQC equation and Q3 equation in the ABS list
as well as the relation between Q3 equation and the `lower' equations Q2, Q1, H3, H2 and H1
(also see \cite{Hietarinta-CAC}).
In Ref.\cite{Nijhoff-ABS} these ABS lattice equations (see Appendix \ref{A:2}) were re-parameterized
so that their solutions can be expressed through Cauchy matrices.
These re-parametrisations are\cite{Nijhoff-ABS}
\begin{equation}
\begin{array}{lll}
Q3: & ~~\mathring{p}=\frac{P}{p^2-a^2}=\frac{p^2-b^2}{P}, &~~\mathring{q}=\frac{Q}{q^2-a^2}=\frac{q^2-b^2}{Q}, \\
Q2,Q1: & ~~\mathring{p}=\frac{a^2}{p^2-a^2}, &~~\mathring{q}=\frac{a^2}{q^2-a^2}, \\
H3: & ~~ \mathring{p}=\frac{P}{a^2-p^2}=\frac{1}{P}, &~~\mathring{q}=\frac{Q}{a^2-q^2}=\frac{1}{Q}, \\
H2,H1: & ~~\mathring{p}=-p^2, &~~\mathring{q}=-q^2.
\end{array}
\end{equation}
And the re-paramaterized lattice equations are
\begin{subequations}
\label{ABS-list-p}
\begin{align}
Q3: &~~P(u\widehat{u}+\widetilde{u}\widehat{\widetilde{u}})
-Q(u\widetilde{u} +\widehat{u}\widehat{\widetilde{u}}) =(p^2-
q^2)\bigg((\widetilde{u}
\widehat{u}+u\widehat{\widetilde{u}})+\frac{\delta^2}{4PQ}\bigg),
\label{Q3-p}\\
Q2: &~~(q^2-a^2)(u-\wh{u})(\wt{u}-\wh{\wt{u}})-(p^2-a^2)(u-\wt{u})(\wh{u}-\wh{\wt{u}})\nonumber\\
&~~~~ + (p^2-a^2)(q^2-a^2)(q^2-p^2)(u+\wt{u}+\wh{u}+\wh{\wt{u}})\nonumber\\
&~~ =(p^2-a^2)(q^2-a^2)(q^2-p^2)[(p^2-a^2)^2+(q^2-a^2)^2-(p^2-a^2)(q^2-a^2)],
\label{Q2-p}\\
Q1: &~~(q^2-a^2)(u-\wh{u})(\wt{u}-\wh{\wt{u}})-
(p^2-a^2)(u-\wt{u})(\wh{u}-\wh{\wt{u}}) =\frac{\delta^2a^4(p^2-q^2)}{(p^2-a^2)(q^2-a^2)},
\label{Q1-p}\\
H3: &~~  P(a^2-q^2)(u\wt{u}+\wh{u}\wh{\wt{u}})-Q(a^2-p^2)
(u\wh{u}+\wt{u}\wh{\wt{u}})= \delta(p^2-q^2), \label{H3-p} \\
H2: &~~ (u-\wh{\wt{u}})(\wt{u}-\wh{u})+(p^2-q^2)
(u+\wt{u}+\wh{u}+\wh{\wt{u}})=p^4-q^4, \label{H2-p}\\
H1: &~~ (u-\wh{\wt{u}})(\wh{u}-\wt{u})=p^2-q^2, \label{H1-p}
\end{align}
\end{subequations}
where in \eqref{Q3-p}  $(p, P)=\mathfrak{p}$ and $(q,Q)=\mathfrak{q}$ are the points on the elliptic curve
\begin{equation}
 \{(x,X)|X^2 =
(x^2-a^2)(x^2-b^2)\},
\label{elliptic-Q3}
\end{equation}
in \eqref{H3-p}
\[P^2=a^2-p^2,~~ Q^2=a^2-q^2,\]
and in Q3 and Q2 the dependent variable $u$ has been scaled by
\[u\to u(b^2-a^2),~~ u\to \frac{a^4 u}{(p^2-a^2)^2 (q^2-a^2)^2},\]
respectively.

\subsubsection{Solutions}

In Ref.\cite{Nijhoff-ABS} solutions to Q3 equation \eqref{Q3-p} is given by means of
the relation between NQC equation and Q3 equation;
solutions for other equations in the list \eqref{ABS-list-p}
are derived by using the degeneration relation
\vspace{-8mm}
\begin{center}
\begin{displaymath}
\xymatrix{ \boxed{Q3} \ar[d] \ar[r] & \boxed{Q2} \ar[d] \ar[r] & \boxed{Q1} \ar[d] \\
\boxed{H3} \ar[r] & \boxed{H2} \ar[r] & \boxed{H1} }
\end{displaymath}
\begin{minipage}{11cm}{\footnotesize~~~~~~~~~~~~~~~~~~~~~~~~~~~~~~~~~
{Fig.1}  Degeneration relation}
\end{minipage}
\end{center}

In the following we list out solution formulae\cite{Nijhoff-ABS} for the lattice equations listed in \eqref{ABS-list-p},
and here we claim that these formulae are valid
for all the solutions $\{\br, \bM\}$ of our starting matrix equation set \eqref{cond-2}
together with corresponding $\Ga$ and $\tc$ which are described in Theorem \ref{T:3}.

For Q3 equation \eqref{Q3-p},
\begin{subequations}
\begin{equation}
u=A\cdot G(a,b)+B\cdot G(a,-b) +C\cdot G(-a,b)+D\cdot G(-a,-b),
\end{equation}
where
\begin{equation}
G(a,b)=\biggl(\frac{P}{(p-a)(p-b)}\biggr)^n \biggl(\frac{Q}{(q-a)(q-b)}\biggr)^m [1-(a+b)S(a,b)],
\end{equation}
$S(a,b)$ is defined in \eqref{SSV-Sab}, and $A,B,C,D$ are arbitrary constants satisfying
\begin{equation}
A\, D(a+b)^2-B\,C(a-b)^2=-\frac{\delta^2}{16 ab}.
\end{equation}
\end{subequations}

For Q2 equation \eqref{Q2-p},
\begin{subequations}
\begin{align}
u=&\gamma(a)\Bigl[\frac{1}{4}(\xi^2+1)+a\xi S(-a,a)+a^2(Z(a,-a)+Z(-a,a))\nonumber\\
&~~~~~~~ +A\, D+\frac{A}{2}\rho(a)(1-2a S(a,a))+\frac{D}{2}\rho(-a)(1+2a S(-a,-a))\Bigr],
\end{align}
where
\begin{align}
& \xi=2a\Bigl(\frac{pn}{a^2-p^2}+\frac{qm}{a^2-q^2}\Bigr)+\xi_0,~~ (\xi_0 \mathrm{~a~ constant}), \label{xi}\\
& \rho(a)=\biggl(\frac{p+a}{p-a}\biggr)^n \biggl(\frac{q+a}{q-a}\biggr)^m \rho_{00},~~ (~\rho_{00} \mathrm{~a~ constant}),\label{rho-a}\\
& \gamma(a)=\frac{a^4 }{(p^2-a^2)^2 (q^2-a^2)^2},
\end{align}
\end{subequations}
$S(a,b)$ and $Z(a,b)$ are defined in \eqref{SSV}, and $A,D$ are arbitrary constants.

For Q1 equation \eqref{Q1-p}, solution with exponential background is
\begin{subequations}
\begin{equation}
u=A\,\rho(a)(1-2a S(a,a))+ D\,\rho(-a)(1+2a S(-a,-a))+B(\xi+2a S(-a,a)),
\end{equation}
where $S(a,b)$ is defined in \eqref{SSV-Sab}, $\xi$ defined in \eqref{xi}, $\rho(a)$ defined in \eqref{rho-a},
and $A,B,D$ are arbitrary constants satisfying
\begin{equation}
16 A\, D+4B^2=\delta^2;
\end{equation}
\end{subequations}
while solution with linear background is
\begin{subequations}
\begin{equation}
u=\delta (\nu^2-2\nu S^{(-1,-1)}+2S^{(-1,-2)})+2A(\nu +c_0-S^{(-1,-1)}),
\end{equation}
where $S^{(i,j)}$ is defined in \eqref{SSV-Sij},
\begin{equation}
\nu=\frac{n}{p}+\frac{m}{q}+\nu_0,~~ (\nu_0 \mathrm{~a~ constant}), \label{nu}
\end{equation}
\end{subequations}
and $A,c_0$ are arbitrary constants.

For H3 equation \eqref{H3-p}, solution is
\begin{subequations}
\begin{equation}
u=(A+(-1)^{n+m}B)\vartheta V(a)+(D+(-1)^{n+m}C)\vartheta^{-1} V(-a),
\end{equation}
where $V(a)$ is defined in \eqref{SSV-V}, $\vartheta$ is defined as
\begin{equation}
\vartheta=\biggl(\frac{P}{a-p}\biggr)^n \biggl(\frac{Q}{a-q}\biggr)^m \vartheta_{00},~~ (~\vartheta_{00} \mathrm{~a~ constant}),
\end{equation}
and $A,B,C,D$ are arbitrary constants satisfying
\begin{equation}
A\, D-B\, C=\frac{-\delta}{4a}.
\end{equation}
\end{subequations}

For H2 equation \eqref{H2-p}, solution is
\begin{subequations}
\begin{equation}
u=\frac{1}{4}\zeta^2-\zeta S^{(0,0)}+2S^{(0,1)}-A^2+(-1)^{n+m}A(\zeta+c_0-2S^{(0,0)}),
\end{equation}
where $S^{(i,j)}$ is defined in \eqref{SSV-Sij}, $\zeta$ is defined as
\begin{equation}
\zeta=pn+qm+\zeta_{0},~~ (~\zeta_{0} \mathrm{~a~ constant}),
\label{zeta}
\end{equation}
and $A,c_0$ are arbitrary constants.
\end{subequations}

For H1 equation \eqref{H1-p}, solution is
\begin{subequations}
\label{H1-u}
\begin{equation}
u=A(\zeta-S^{(0,0)})+(-1)^{n+m}B(\zeta+c_0-2S^{(0,0)}),
\end{equation}
where $S^{(i,j)}$ is defined in \eqref{SSV-Sij}, $\zeta$   defined in \eqref{zeta},
$A,B, c_0$ are arbitrary constants and $A,B$ satisfy
\begin{equation}
A^2-B^2=1.
\end{equation}
\end{subequations}

Here we note that by the transformation
\[w=\zeta-u,\]
the lpKdV equation \eqref{lpKdV-a} is transformed to H1 equation \eqref{H1-p}.
This corresponds to $A=1, B=0$ in \eqref{H1-u}.
Besides, we do not consider A1 and A2 equation in the ABS list  due to simple relations with Q1 and Q3($\delta=0$) equation.

\section{Conclusions}
\label{sec:6}

The usual Cauchy matrix approach starts from known plain wave factor vector $\br$ and known dressed Cauchy matrix $\bM$.
In the paper we have described a generalized Cauchy matrix approach
where the starting point is a matrix equation set with unknown $\br$ and $\bM$.
Such a starting equation set can equally lead to those recurrence relations for
the defined scalar functions $S^{(i,j)}$ and $S(a,b)$.
With regard to solutions, we showed that $\bK$ in the starting equation set provides same solutions
for the related lattice equations as its canonical form $\Ga$ dose.
This enables us to simplify the starting equation set and investigate its solutions
according to the eigenvalue structure of the canonical matrix $\Ga$.
We obtained explicit forms of $\br$ and $\bM$ for a generic canonical matrix $\Ga$.
Besides, we proved the symmetric property of $S^{(i,j)}=S^{(j,i)}$,
which is crucial to generating closed recurrence relations (lattice equations).
As applications, based on Ref.\cite{Nijhoff-ABS} we obtained solution formulae for many
lattice equations,
such as the lpKdV, lpmKdV, lSKdV, NQC, Q3, Q2, Q1, H3, H2 and H1 equation.
Since the starting  equation set admits more choices for $\br$ and $\bM$,
solutions of the above lattice equations are more general than usual solitons.
In fact, in some sense, solutions corresponding to $\Ga$ being a Jordan block
can be considered as a limit result of the case of $\Ga$ being diagonal (cf.\cite{Zhang-KdV,Hietarinta-Bou,Zhang-H1-G}),
or multiple-pole solutions in direct linearization approach.
Besides, usually, zero eigenvalues lead to rational solutions\cite{Zhang-H1Q3-R}.
However, since we need $\Ga$ to satisfy invertible conditions (see Proposition \ref{P:1}),
eigenvalues of $\Ga$ can not be zero and therefore here the obtained solutions do not include rational solutions.

Finally, let us go back to the starting equation set \eqref{cond-1} or \eqref{cond-2},
which consists of three equations.
Reviewing the generalized Cauchy matrix approach, we can see that
the first two equations are used to determine the plain wave factor vector $\br$.
Actually, these two equations include implicitly 2nd-order dispersion relation,
which determines the obtained lattice equations are of the KdV-type (cf.\cite{Nijhoff-ABS}).
If the  dispersion relation is of 3rd-order, one can get lattice Boussinesq type equations\cite{Walker-2011,Tongas,ZZS-2012,FZZ-2012}.
The third equation in the starting equation set is used to define $\bM$.
To find the explicit form of $\bM$ we write it in the dressed form $\bM=\bF\bG\bH$.
This works as well as in the proof of the symmetric property $S^{(i,j)}=S^{(j,i)}$.
Besides, some algebraic properties and technique also make the discussions and proofs neat.

\vskip 20pt
\section*{Acknowledgments}
The paper is originated when the author (Zhang) participated in the DIS-2009 Programme
in the Isaac Newton Institute for Mathematical Sciences.
We are very grateful to Prof. Nijhoff for enthusiastic guiding and discussing in discrete integrable systems.
Our thanks are also extended to Prof. Aktosun for kindly reading the manuscript.
This project is supported by the NSF of China
(No. 11071157), SRF of the DPHE of China (No. 20113108110002),
Shanghai Leading Academic Discipline Project (No. J50101).

\begin{appendix}
\section{Proof of Lemma \ref{L:1-1}}\label{A:1}
There is a nonsingular matrix $T\in \mathbb{C}_{N\times N}$
such that
\begin{equation}
\Ga=T\bK T^{-1},
\end{equation}
where $\Ga=(\gamma_{ij})_{N\times N}$ is an upper triangular matrix with $\gamma_{ii}=k_i$ for $i=1,2,\cdots,N$.
We also denote
\[\bB=T\bA T^{-1}=(\beta_1, \beta_2,\cdots,\beta_N),\]
where $\{\beta_j\}$ are column vectors of $\bB$.
Then \eqref{anti-c} can be rewritten as
\begin{equation}
\Ga \bB+\bB\Ga=\mathbf{0}.
\label{anti-b}
\end{equation}
The first column of \eqref{anti-b} reads
\[(k_1\bI +\Ga)\beta_1=\mathbf{0}.\]
Noting that $|k_1\bI +\Ga|\neq 0$ due to the condition \eqref{kij},
we have $\beta_1=\mathbf{0}$.
Then, with this in hand the second column of \eqref{anti-b} is simplified to
\[(k_2\bI +\Ga)\beta_2=\mathbf{0},\]
which yields $\beta_2=0$.
Repeating this procedure step by step we can successively get
$\beta_3=\beta_4=\cdots=\beta_N=\mathbf{0}$.
This means $\bB=\mathbf{0}$ and then $\bA=\mathbf{0}$.

\section{ABS lattice equations}\label{A:2}

Here we list out lattice equations in ABS list\cite{ABS-2003}. These equations are
\begin{subequations}
\label{ABS-list}
\begin{align}
Q4:~~ &\mathring{p}(u\wt{u}+\wh{u}\wh{\wt{u}})-
\mathring{q}(u\wh{u}+\wt{u}\wh{\wt{u}})\nonumber\\
&~~~~~~~~~~~~~=\frac{\mathring{p}\mathring{Q}-\mathring{q}\mathring{P}}{1-\mathring{p}^2\mathring{q}^2}
\big((\wt{u}\wh{u}+u\wh{\wt{u}})-
\mathring{p}\mathring{q}(1+u\wt{u}\wh{u}\wh{\wt{u}})\big),\label{Q4} \\
Q3:~~ & \mathring{p}(1-\mathring{q}^2)(u\wh{u}+\wt{u}\wh{\wt{u}})-
\mathring{q}(1-\mathring{p}^2)(u\wt{u}+\wh{u}\wh{\wt{u}})\nonumber\\
&~~~~~~~~~~~~~=(\mathring{p}^2-\mathring{q}^2)\bigg((\wt{u}\wh{u}+u\wh{\wt{u}})+
\delta^2\frac{(1-\mathring{p}^2)(1-\mathring{q}^2)}{4\mathring{p}\mathring{q}}\bigg),\label{Q3}\\
Q2:~~ & \mathring{p}(u-\wh{u})(\wt{u}-\wh{\wt{u}})-
\mathring{q}(u-\wt{u})(\wh{u}-\wh{\wt{u}})
+\mathring{p}\mathring{q}(\mathring{p}-\mathring{q})(u+\wt{u}+\wh{u}+\wh{\wt{u}})\nonumber\\
&~~~~~~~~~~~~~=\mathring{p}\mathring{q}
(\mathring{p}-\mathring{q})(\mathring{p}^2-\mathring{p}\mathring{q}+\mathring{q}^2),\label{Q2}\\
Q1:~~ &  \mathring{p}(u-\wh{u})(\wt{u}-\wh{\wt{u}})-
\mathring{q}(u-\wt{u})(\wh{u}-\wh{\wt{u}})
=\delta^2\mathring{p}\mathring{q} (\mathring{q}-\mathring{p}),
\label{Q1}\\
H3:~~ & \mathring{p}(u\wt{u}+\wh{u}\wh{\wt{u}})-
\mathring{q}(u\wh{u}+\wt{u}\wh{\wt{u}})=\delta(\mathring{q}^2-\mathring{p}^2), \label{H3} \\
H2:~~ & (u-\wh{\wt{u}})(\wt{u}-\wh{u})=(\mathring{p}-\mathring{q})
(u+\wt{u}+\wh{u}+\wh{\wt{u}})+\mathring{p}^2-\mathring{q}^2, \label{H2}\\
H1:~~ & (u-\wh{\wt{u}})(\wt{u}-\wh{u})= \mathring{p}-\mathring{q},
\label{H1}\\
A2:~~ & \mathring{p}(1-\mathring{q}^2)(u\wt{u}+\wh{u}\wh{\wt{u}})-
\mathring{q}(1-\mathring{p}^2)(u\wh{u}+\wt{u}\wh{\wt{u}})
-(\mathring{p}^2-\mathring{q}^2)\big(1+u\wt{u}\wh{u}\wh{\wt{u}}\big)=0, \label{A2} \\
A1:~~ & \mathring{p}(u+\wh{u})(\wt{u}+\wh{\wt{u}})-
\mathring{q}(u+\wt{u})(\wh{u}+\wh{\wt{u}})
=\delta^2\mathring{p}\mathring{q} (\mathring{p}-\mathring{q}),
\label{A1}
\end{align}
\end{subequations}
where $\mathring{p}, \mathring{q}$ are spacing parameters, $\delta$ is an arbitrary constant,
and in Q4 equation $(\mathring{p}, \mathring{P})$ and $(\mathring{q}, \mathring{Q})$
are points on the elliptic curve\cite{Hietarinta-CAC}
\[\{(x,X)|  X^2=x^4-\gamma x^2+1\}.\]

\end{appendix}

\end{document}